%% file: geom.tex
\documentclass[a4paper,10pt]{article}
\usepackage{amsmath}
\usepackage{amssymb}
\usepackage{amsthm}
\usepackage[utf8]{inputenc}
\usepackage{pgf,tikz}
\usepackage{mathrsfs}
\usepackage{titlesec}
\usepackage[bottom]{footmisc}
\usetikzlibrary{arrows}

\titleformat{\section}{\large\bfseries}{\thesection}{1em}{}
\definecolor{cqcqcq}{rgb}{0.6,0.6,0.6}
\newtheorem{thm}{Theorem}[section]
\newtheorem{prop}[thm]{Proposition}
\newtheorem{lemma}[thm]{Lemma}
\newtheorem{coro}[thm]{Corollary}
\newtheorem{definition}[thm]{Definition}
\theoremstyle{definition}
\newtheorem{remark}[thm]{Remark}

\begin{document}

\begin{center}
{\LARGE{Coble's group and the integrability of the Gosset-Elte polytopes and tessellations}}
\vskip6mm

\noindent
{\large James Atkinson\footnote{contact: james.l.atkinson@gmail.com} } \\

\noindent
{\large 08.03.17}\\

\end{center}
\noindent
{\bf Abstract.}
This paper considers the planar figure of a combinatorial polytope or tessellation identified by the Coxeter symbol $k_{i,j}$, inscribed in a conic, satisfying the geometric constraint that each octahedral cell has a centre.
This realisation exists, and is movable, on account of some constraints being satisfied as a consequence of the others. 
A close connection to the birational group found originally by Coble in the different context of invariants for sets of points in projective space, allows to specify precisely a determining subset of vertices that may be freely chosen.
This gives a unified geometric view of certain integrable discrete systems in one, two and three dimensions.
Making contact with previous geometric accounts in the case of three dimensions, it is shown how the figure also manifests as a configuration of circles generalising the Clifford lattices, and how it can be applied to construct the spatial point-line configurations called the Desargues maps.
\vskip10mm

\section{Introduction}
The six-point multi-ratio equation appears in the theory of integrable systems 
on a discrete domain with octahedral cells \cite{DN,Bog,ks1,slg,ks,DDM,ABSo,SKP}.
It has the following geometric meaning, which has not been considered previously in this area.
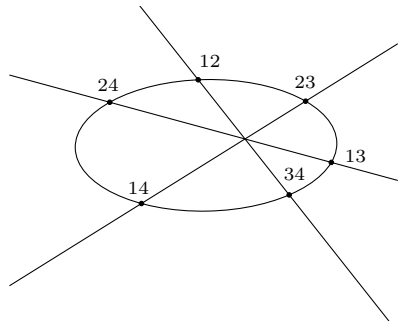
\begin{figure}[b!]
\begin{center}
\input{skp-pic2.tex}
\end{center}
\caption{Concurrent secants of a conic.}
\label{skp-illustration2}
\end{figure}
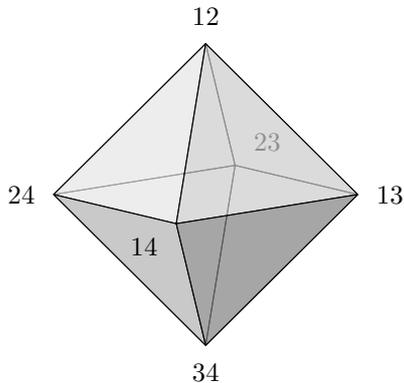
\begin{figure}[t]
\begin{center}
\input{octahedron.tex}
\end{center}
\caption{Octahedron with labelled vertices.}
\label{octahedron}
\end{figure}
\begin{lemma}\label{gg}
If three lines meet a rationally parameterised conic at points corresponding to the three pairs of parameter values,
\begin{equation}\label{stencil}
 \{x_{12},x_{34}\},  \{x_{13},x_{24}\}, \{x_{14},x_{23}\}, 
\end{equation}
then the concurrence of the lines, Figure \ref{skp-illustration2}, is expressed analytically as 
\begin{equation}
\frac{(x_{12}-x_{24})(x_{13}-x_{23})(x_{14}-x_{34})}{(x_{12}-x_{23})(x_{14}-x_{24})(x_{13}-x_{34})}=1.\label{skp}
\end{equation}
\end{lemma}
The calculation to verify this is straightforward, but a more affecting argument in the moduli space of quadratic polynomials, given by Dolgachev \cite{Dolg} (Proposition  9.4.9), is also provided in appendix \ref{roots}.
Correspondence with the octahedron is shown in Figure \ref{octahedron}.

The results of applying Lemma \ref{gg} are consistent with the observations of Adler \cite{AdlInc} on other multi-ratio expressions.
In those cases, there is an interesting connection between integrability and the generalisations of Pascal's theorem due to M\"obius.
The situation here is similar, but involves a different extension of Pascal's figure, described in Section \ref{Pascal}.

The main result in Section \ref{Gosset-Elte} uses Lemma \ref{gg} to associate a combinatorial polytope or tessellation identified by the Coxeter symbol $k_{i,j}$, inscribed in a conic, with a generalised form of Coble's birational group; the group is defined in Section \ref{Coble}.
The points on the conic become points of a circle pattern when the conic is viewed as a model of the inversive plane. 
This generalises the previous geometric view of equation (\ref{skp}) established by Konopelchenko, Schief and King \cite{ks1,ks2}, more precisely their circle-pattern is recovered in the case $k=0$ here.
The connection is explained in Section \ref{Clifford}.

A planar incidence-geometry view of Coble's group has been established by Kajiwara, Masuda, Noumi, Ohta and Yamada \cite{10E9}, with particular attention to the solution in the case $1_{5,2}$ in terms of elliptic functions, and its equivariant extension to the $1_{6,2}$ case, that corresponds to the elliptic Painlev\'e equation \cite{sc}.
The view established here treats uniformly all cases $k_{i,j}$, but turns out to be especially natural with-respect-to a simpler class of solutions that are rational; these are described in Section \ref{solutions}.

The Desargues maps introduced by Doliwa \cite{DDM} are combinatorially rich point-line configurations that are naturally considered in projective space of any dimension $M>1$, but, that are considered here in a restricted commutative and affine setting, $\mathbb{A}^M$.
In this case, they are shown in Section \ref{Desargues} to separate into $M$ independent systems each equivalent to case $k=0$ of Coble's group,
and this is applied to establish the natural determining-set for the configurations. 
Throughout this paper, the projective or affine space is defined over a field, the case of a skew-field is discussed briefly at the end of Section \ref{Desargues}, where it is argued that the initial value problem and therefore the diagonalisation of the Desargues maps, is applicable to the general projective case, but raising this to the level of a proof is outside the scope of this work.

\section{Birational group}\label{Coble}
Coble's group \cite{COBI,COBII} has the following generalisation, connecting it with equation (\ref{skp}).
\begin{definition}[\cite{SKP}]\label{actions}
Let integers $i,j$ be positive and $k$ be non-negative. 
Introduce actions on the arrays of variables 
\begin{equation}
\left[\begin{array}{c}
y_0\\
y_1\\
\vdots\\
y_k
\end{array}\right],
\quad
\left[\begin{array}{cccc}
y_{00} & y_{10} & \cdots & y_{i0}\\
y_{01} & y_{11} & \cdots & y_{i1}\\
\vdots & \vdots & & \vdots \\
y_{0j} & y_{1j} & \cdots & y_{ij}
\end{array}\right],\label{Xdef}
\end{equation}
as follows:
\begin{equation}\label{rational_actions}
\begin{array}{rll}
t_m: & y_{(m-1)n} \leftrightarrow y_{mn}, & m \in \{1,\ldots,i\},\  n\in\{0,\ldots,j\},\\
s_n: & y_{m(n-1)} \leftrightarrow y_{mn}, & m \in \{0,\ldots,i\},\  n\in\{1,\ldots,j\},\\
r_n: & y_{n-1} \leftrightarrow y_{n}, & n\in\{1,\ldots,k\},\\
r_0: & y_0\leftrightarrow y_{00}, \ y_{mn}\rightarrow \bar{y}_{mn}, \ & m\in\{1,\ldots,i\},\  n\in\{1,\ldots,j\},
\end{array}
\end{equation}
where trivial actions are omitted, and $\bar{y}_{mn}$ is determined by the six-point multi-ratio equation imposed on variables
\begin{equation}\label{sixvars}
\{y_0,\bar{y}_{mn}\},\{y_{0n},y_{m0}\},\{y_{00},y_{mn}\},
\end{equation}
cf. Lemma \ref{gg}, i.e., 
\begin{equation}\label{yeqn}
\frac{(y_0-y_{m0})(y_{0n}-y_{mn})(y_{00}-\bar{y}_{mn})}{(y_0-y_{mn})(y_{00}-y_{m0})(y_{0n}-\bar{y}_{mn})}=1.
\end{equation}
\end{definition}

\begin{figure}
\begin{center}
\begin{tikzpicture}[thick]
  \node at (-3,0) [draw,circle,fill=white,minimum size=8pt,inner sep=0pt]{};
  \tikzstyle{every node}=[draw,circle,fill=black,minimum size=4pt,inner sep=0pt];
  \node (r0) at (0,0) [label={[label distance=2pt]below:$r_0$}]{};
  \node (r1) at (-1.2,0) [label={[label distance=2pt]below:$r_1$}]{};
  \node (rk) at (-3,0) [label={[label distance=2pt]below:$r_k$}]{};
  \node (s1) at (1,0.8) [label={[label distance=2pt]above:$s_1$}]{};
  \node (t1) at (1,-0.8) [label={[label distance=2pt]below:$t_1$}]{};
  \node (sj) at (2.8,0.8) [label={[label distance=2pt]above:$s_j$}]{};
  \node (ti) at (2.8,-0.8) [label={[label distance=2pt]below:$t_i$}]{};
\draw[dashed] (rk)--(r1);
\draw (r1)--(r0)--(s1);
\draw (r0)--(t1);
\draw[dashed] (s1)--(sj);
\draw[dashed] (t1)--(ti);
\draw (r1)--(-1.6,0);
\draw (rk)--(-2.6,0);
\draw (s1)--(1.4,0.8);
\draw (sj)--(2.4,0.8);
\draw (t1)--(1.4,-0.8);
\draw (ti)--(2.4,-0.8);
\end{tikzpicture}
\end{center}
\caption{Coxeter graph with nodes corresponding to actions of Definition \ref{actions}.}
\label{cdd}
\end{figure}
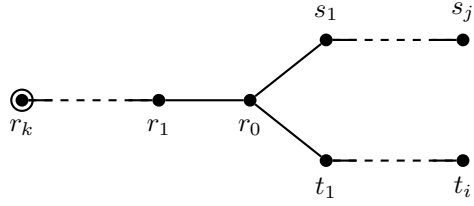
The generators (\ref{rational_actions}) satisfy relations encoded in the Coxeter graph of Figure \ref{cdd}.
Specifically, the group relations correspond to identities in the field of rational functions in the variables (\ref{Xdef}).
The original group of Coble can be viewed as a partial integration of this one that is available if $k=1$.

To give the description of this group afforded by Lemma \ref{gg}, is the principal aim of this paper. 

\section{Extension of Pascal's hexagon}\label{Pascal}
A suitable starting point is an elementary but intriguing extension to Pascal's figure of a hexagon inscribed in a conic. 
\begin{prop}\label{direct}
Consider six points on a conic, $C$, labelled as follows:
\begin{equation} \label{exhex}
p_{\{1,3\}},p_{\{2,4\}},p_{\{1,5\}},p_{\{2,3\}},p_{\{1,4\}},p_{\{2,5\}}.
\end{equation}
Use $a+b$ to denote the line determined by two points.
By Pascal's theorem, the points
\begin{equation}\label{pascal-points}
\begin{split}
p_3:=(p_{\{1,4\}}+p_{\{2,5\}})\cap(p_{\{2,4\}}+p_{\{1,5\}}),\\
p_4:=(p_{\{1,3\}}+p_{\{2,5\}})\cap(p_{\{2,3\}}+p_{\{1,5\}}),\\
p_5:=(p_{\{1,3\}}+p_{\{2,4\}})\cap(p_{\{2,3\}}+p_{\{1,4\}}),
\end{split}
\end{equation}
are collinear, determining a line $\pi$,
\begin{equation}\label{pascal-line}
p_3+p_4 = p_4+p_5 = p_5+p_3 =: \pi. 
\end{equation}

Add one further point to the figure, $p_{\{1,2\}}$, chosen freely on $C$. 
Lines connecting this point with the determining points of the Pascal line (\ref{pascal-points}), intersect $C$ at three further points:
\begin{equation}\label{augment}
\begin{split}
p_{\{4,5\}}:=(p_{\{1,2\}}+p_3)\cap(C \setminus p_{\{1,2\}}),\\
p_{\{3,5\}}:=(p_{\{1,2\}}+p_4)\cap(C \setminus p_{\{1,2\}}),\\
p_{\{3,4\}}:=(p_{\{1,2\}}+p_5)\cap(C \setminus p_{\{1,2\}}).\\
\end{split}
\end{equation}
The following incidences then occur, determining two final points on $\pi$:
\begin{equation}\label{new-points}
\begin{split}
\pi\cap (p_{\{2,3\}}+p_{\{4,5\}}) = \pi \cap (p_{\{2,4\}}+p_{\{3,5\}}) = \pi \cap (p_{\{2,5\}}+p_{\{3,4\}}) =: p_1,\\ 
\pi\cap (p_{\{1,3\}}+p_{\{4,5\}}) = \pi \cap (p_{\{1,5\}}+p_{\{3,4\}}) = \pi \cap (p_{\{1,4\}}+p_{\{3,5\}}) =: p_2.\\ 
\end{split}
\end{equation}
See Figure \ref{conic}.
\end{prop}

\begin{figure}[t]
{\input{conic_theorem.tex}}
\caption{
Illustration of Proposition \ref{direct},
points $p_{\{\alpha,\beta\}}$ and $p_\gamma$ are labelled as $\alpha\beta$ and $\gamma$ respectively.
The combinatorial symmetries of this figure correspond to free permutation of the five indices.
}
\label{conic}
\end{figure}
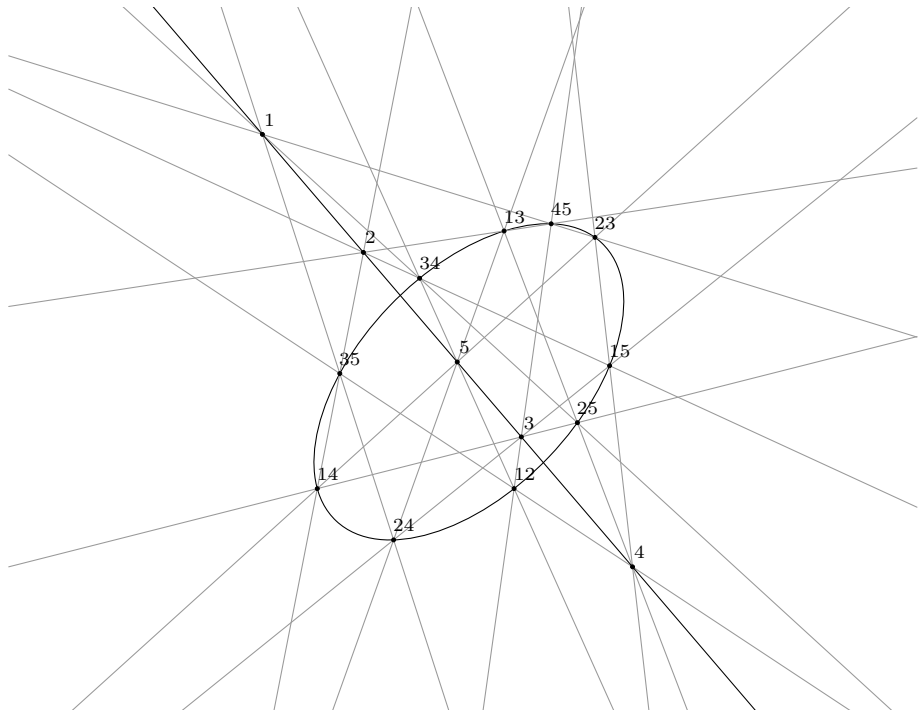
\begin{proof}
Pascal's theorem can be applied to demonstrate that 
\begin{equation}\label{showing}
\pi\cap (p_{\{2,3\}}+p_{\{4,5\}}) = \pi \cap (p_{\{2,4\}}+p_{\{3,5\}}),
\end{equation}
and the other equalities in (\ref{new-points}) are established in the same way.
Consider the inscribed hexagon whose vertices are, in consecutive order,
\begin{equation}\label{hex2}
p_{\{1,2\}},p_{\{3,5\}},p_{\{2,4\}},p_{\{1,5\}},p_{\{2,3\}},p_{\{4,5\}}.
\end{equation}
Because of how $p_{\{1,2\}}$, $p_{\{4,5\}}$ and $p_{\{3,5\}}$ have been defined, two of the three determining points for the Pascal line of this hexagon are already known, they are
\begin{equation}
\begin{split}
(p_{\{1,2\}}+p_{\{4,5\}})\cap(p_{\{2,4\}}+p_{\{1,5\}})=p_3,\\
(p_{\{1,2\}}+p_{\{3,5\}})\cap(p_{\{2,3\}}+p_{\{1,5\}})=p_4.\\
\end{split}
\end{equation}
Therefore, it is common with the original Pascal line (\ref{pascal-line}), and the third determining point from hexagon (\ref{hex2}) must also be somewhere on it, i.e.,
\begin{equation}
(p_{\{3,5\}}+p_{\{2,4\}})\cap(p_{\{2,3\}}+p_{\{4,5\}})\in\pi,\\
\end{equation}
confirming (\ref{showing}).
\end{proof}

That the actions in the case $(i,j,k)=(2,1,0)$ of Definition \ref{actions} satisfy relations encoded in the corresponding Coxeter graph (Figure \ref{cdd}), can be regarded as an analytic formulation of Proposition \ref{direct}. 
Identify the elements of the arrays (\ref{Xdef}) with points on $C$ as follows,
\begin{equation}\label{Sdef}
\left[\begin{array}{c}
p_{\{1,2\}}\\
\end{array}\right],
\quad
\left[\begin{array}{cccc}
p_{\{1,3\}} & p_{\{1,4\}} & p_{\{1,5\}}\\
p_{\{2,3\}} & p_{\{2,4\}} & p_{\{2,5\}}
\end{array}\right],
\end{equation}
and the actions (\ref{rational_actions}) with the permutations of indices that generate the combinatorial symmetries of the figure,
\begin{equation}\label{indact}
s_1: 1\leftrightarrow 2, \
r_0: 2\leftrightarrow 3, \
t_1: 3\leftrightarrow 4, \
t_2: 4\leftrightarrow 5.
\end{equation}
The array on the right in (\ref{Sdef}) corresponds to the vertices of the original hexagon (\ref{exhex}), and on the left, to the added point, $p_{\{1,2\}}$.
The actions $s_1$, $t_1$ and $t_2$ in (\ref{indact}) generate the subgroup of combinatorial symmetries of the initial hexagon, and the induced action on the array on the right in (\ref{Sdef}) is simply to permute the rows and columns, confirming identification with the actions in (\ref{rational_actions}).
The index permutation $2\leftrightarrow 3$, $r_0$ in (\ref{indact}), induces transposition of the first array elements in (\ref{Sdef}), but it induces a set of geometric operations on entries corresponding to $p_{\{2,4\}}$ and $p_{\{2,5\}}$. 
The new points, $p_{\{3,4\}}$ and $p_{\{3,5\}}$, are determined from the given ones by (\ref{pascal-points}) and (\ref{augment}).
Performing these operations analytically using Lemma \ref{gg}, by identifying points on $C$ with corresponding values of a parameter, confirms the action of $r_0$ listed in Definition \ref{actions}.

This case illustrates the general situation. 
The arrays (\ref{Xdef}) correspond to the set of points on a conic from which a figure is determined, and elements of the group generated by (\ref{rational_actions}) obtain the image of this determining-set under corresponding combinatorial symmetries of the figure.
In the case just described the following can be checked by inspection.

\begin{remark}\label{strong}
With regard to the group of combinatorial symmetries of Figure \ref{conic}, each point is contained in some image of the determining-set, and the subgroup that fixes the determining-set point-wise, is trivial.
This means the figure can be recovered from the determining-set using the birational group, and the birational group faithfully represents the combinatorial symmetries of the figure.
\end{remark}

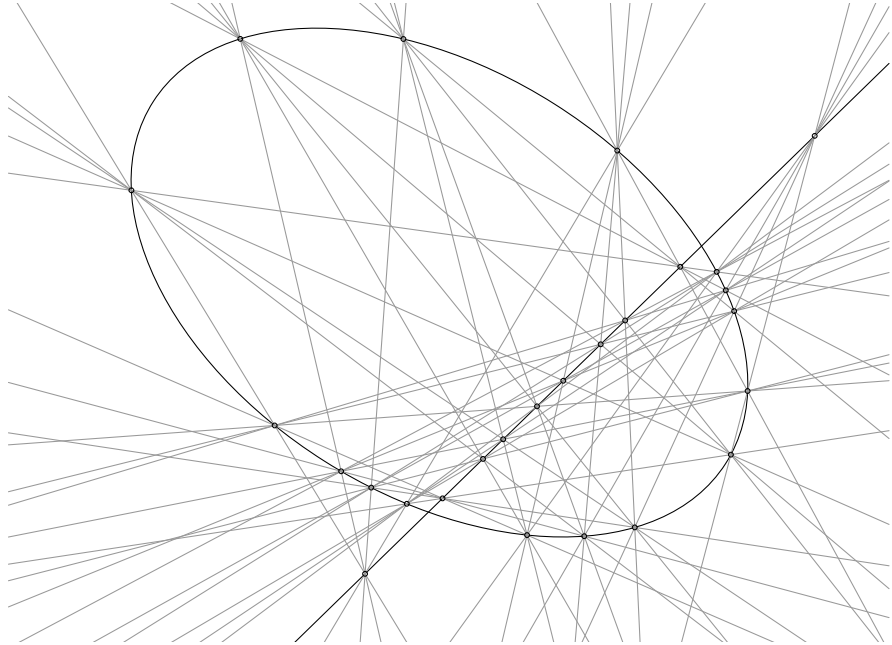
\begin{figure}[t!]
{\input{sixteen.tex}}
\caption{
Example in the sequence of extensions of Pascal's hexagon corresponding to row three of Table \ref{sequence}.
}
\label{sixteen}
\end{figure}
\begin{table}[t!]
\begin{center}
\begin{tabular}{llll}
& Points on $C$ &  Points on $\pi$ & Pascal sub-figures\\
\hline
\rule{0pt}{3ex}
\parbox[c,c]{72pt}{
\begin{tikzpicture}
\tikzstyle{every node}=[draw,circle,fill=black,minimum size=2pt,inner sep=0pt];
\node (4) at (1.2,0.2) {};
\node (5) at (1.2,-0.2) {};
\node (6) at (1.6,-0.2) {};
\draw (4);
\draw (5)--(6);
\end{tikzpicture}
} &
$6\ [2]$ & $3\ [2]$ & $1$
\\
\rule{0pt}{3ex}
\parbox[c,c]{72pt}{
\begin{tikzpicture}
\tikzstyle{every node}=[draw,circle,fill=black,minimum size=2pt,inner sep=0pt];
\node at (1.2,0) [draw,circle,fill=white,minimum size=4pt,inner sep=0pt]{};
\node (4) at (1.2,0) {};
\node (5) at (1.6,0.2) {};
\node (6) at (1.6,-0.2) {};
\node (7) at (2.0,-0.2) {};
\draw (4)--(5);
\draw (4)--(6)--(7);
\end{tikzpicture}
} &
$10\ [3]$ & $5\ [3]$ & $10$
\\
\rule{0pt}{3ex}
\parbox[c,c]{72pt}{
\begin{tikzpicture}
\tikzstyle{every node}=[draw,circle,fill=black,minimum size=2pt,inner sep=0pt];
\node at (0.8,0) [draw,circle,fill=white,minimum size=4pt,inner sep=0pt]{};
\node (3) at (0.8,0) {};
\node (4) at (1.2,0) {};
\node (5) at (1.6,0.2) {};
\node (6) at (1.6,-0.2) {};
\node (7) at (2.0,-0.2) {};
\draw (3)--(4)--(5);
\draw (4)--(6)--(7);
\end{tikzpicture}
} &
$16 \ [5]$ & $10\ [4]$ & $80$
\\
\rule{0pt}{3ex}
\parbox[c,c]{72pt}{
\begin{tikzpicture}
\tikzstyle{every node}=[draw,circle,fill=black,minimum size=2pt,inner sep=0pt];
\node at (0.4,0) [draw,circle,fill=white,minimum size=4pt,inner sep=0pt]{};
\node (2) at (0.4,0) {};
\node (3) at (0.8,0) {};
\node (4) at (1.2,0) {};
\node (5) at (1.6,0.2) {};
\node (6) at (1.6,-0.2) {};
\node (7) at (2.0,-0.2) {};
\draw (2)--(3)--(4)--(5);
\draw (4)--(6)--(7);
\end{tikzpicture}
} &
$27\ [10]$ & $27\ [5]$ & $720$
\\
\rule{0pt}{3ex}
\parbox[c,c]{72pt}{
\begin{tikzpicture}
\tikzstyle{every node}=[draw,circle,fill=black,minimum size=2pt,inner sep=0pt];
\node at (0,0) [draw,circle,fill=white,minimum size=4pt,inner sep=0pt]{};
\node (1) at (0,0) {};
\node (2) at (0.4,0) {};
\node (3) at (0.8,0) {};
\node (4) at (1.2,0) {};
\node (5) at (1.6,0.2) {};
\node (6) at (1.6,-0.2) {};
\node (7) at (2.0,-0.2) {};
\draw (1)--(2)--(3)--(4)--(5);
\draw (4)--(6)--(7);
\end{tikzpicture}
} &
$56\ [27]$ & $126\ [6]$ & $10080$
\\
\rule{0pt}{3ex}
\parbox[c,c]{72pt}{
\begin{tikzpicture}
\tikzstyle{every node}=[draw,circle,fill=black,minimum size=2pt,inner sep=0pt];
\node at (-0.4,0) [draw,circle,fill=white,minimum size=4pt,inner sep=0pt]{};
\node (0) at (-0.4,0) {};
\node (1) at (0,0) {};
\node (2) at (0.4,0) {};
\node (3) at (0.8,0) {};
\node (4) at (1.2,0) {};
\node (5) at (1.6,0.2) {};
\node (6) at (1.6,-0.2) {};
\node (7) at (2.0,-0.2) {};
\draw (0)--(1)--(2)--(3)--(4)--(5);
\draw (4)--(6)--(7);
\end{tikzpicture}
} &
$240\ [126]$ & $2160\ [7]$ & $483840$
\\
\end{tabular}
\caption{
Enumerative description of the finite figures related to case $(i,j)=(2,1)$ of Definition \ref{actions}.
The first row corresponds to a single hexagon, and is a case ($k=-1$) not included in Definition \ref{actions}.
The second row corresponds to Figure \ref{conic}, and the third row to Figure \ref{sixteen}.
The number of lines (not including $\pi$) through each point is included in brackets.
}
\label{sequence}
\end{center}
\end{table}
Loosely speaking, incrementing $k$ iterates the procedure described in Proposition \ref{direct} by adding, at each step, a further freely chosen point on $C$, joining it to the previously marked points of $\pi$, and then adding more lines and points until the figure is symmetric.
For example, Figure \ref{sixteen} shows the next in the sequence after Figure \ref{conic}. 
After some iterations the symmetrisation procedure fails to terminate, corresponding to when the associated reflection group is affine.
Some details of this sequence are given in Table \ref{sequence}.

On the other hand, larger values of $i$ and $j$, due to the possible choices of three columns and two rows, or two columns and three rows, from the array on the right in (\ref{Xdef}), mean the determining-set involves multiple Pascal figures each with a different Pascal line.
Under the group action, these Pascal lines form a configuration $\Pi$, the cases of finite configurations are listed in Table \ref{main}.
The affine cases lead to configurations with a finite number of lines through each point and points on each line, but which are unbounded in number of points and lines.

The combinatorial description of the general figure is obtained via a connection with certain uniform polytopes and tessellations.
\begin{table}[t!]
\begin{center}
\begin{tabular}{lll}
&Points on $C$&Configuration $\Pi$\\
\hline
\rule{0pt}{3ex}
\parbox[c,c]{72pt}{
\begin{tikzpicture}
\tikzstyle{every node}=[draw,circle,fill=black,minimum size=2pt,inner sep=0pt];
\node at (0,0) [draw,circle,fill=white,minimum size=4pt,inner sep=0pt]{};
\node (1) at (0,0) {};
\node (2) at (0.4,0.2) {};
\node (3) at (1.0,0.2) {};
\node (4) at (0.4,-0.2) {};
\node (5) at (1.0,-0.2) {};
\draw (4)--(1)--(2);
\draw[dashed] (4)--(5);
\draw[dashed] (2)--(3);
\end{tikzpicture}
} &
$\frac{(i+j+2)!}{(i+1)!(j+1)!}$ $\left[\frac{i(i+1)j(j+1)}{4}\right]$ & $(p_{i+j-2},q_5)$, $p=\frac{(i+j+2)!}{(i-1)!(j-1)!4!}$
\\
\rule{0pt}{3ex}
\parbox[c,c]{72pt}{
\begin{tikzpicture}
\tikzstyle{every node}=[draw,circle,fill=black,minimum size=2pt,inner sep=0pt];
\node at (0,0) [draw,circle,fill=white,minimum size=4pt,inner sep=0pt]{};
\node (1) at (0,0) {};
\node (2) at (0.4,0) {};
\node (4) at (0.8,0.2) {};
\node (3) at (0.8,-0.2) {};
\node (5) at (1.4,0.2) {};
\draw (1)--(2)--(3);
\draw (2)--(4);
\draw [dashed] (4)--(5);
\end{tikzpicture}
} &
$2^{j+2}$ $\left[\frac{(j+3)!}{(j-1)!4!}\right]$ & $(p_{j-1},q_{10})$, $p=\frac{2^{j-1}(j+3)!}{(j-1)!4!}$
\\
\rule{0pt}{3ex}
\parbox[c,c]{72pt}{
\begin{tikzpicture}
\tikzstyle{every node}=[draw,circle,fill=black,minimum size=2pt,inner sep=0pt];
\node at (0,0) [draw,circle,fill=white,minimum size=4pt,inner sep=0pt]{};
\node (1) at (0,0) {};
\node (2) at (0.4,0) {};
\node (3) at (0.8,0) {};
\node (4) at (1.2,0.2) {};
\node (5) at (1.2,-0.2) {};
\node (6) at (1.6,-0.2) {};
\draw (1)--(2)--(3)--(4);
\draw (3)--(5)--(6);
\end{tikzpicture}
} &
$27$ $[10]$ & $(27_1,1_{27})$
\\
\rule{0pt}{3ex}
\parbox[c,c]{72pt}{
\begin{tikzpicture}
\tikzstyle{every node}=[draw,circle,fill=black,minimum size=2pt,inner sep=0pt];
\node at (0,0) [draw,circle,fill=white,minimum size=4pt,inner sep=0pt]{};
\node (1) at (0,0) {};
\node (2) at (0.4,0) {};
\node (3) at (0.8,0) {};
\node (4) at (1.2,0) {};
\node (5) at (1.6,0.2) {};
\node (6) at (1.6,-0.2) {};
\node (7) at (2.0,-0.2) {};
\draw (1)--(2)--(3)--(4)--(5);
\draw (4)--(6)--(7);
\end{tikzpicture}
} &
$56$ $[27]$ & $(126_1,1_{126})$
\\
\rule{0pt}{3ex}
\parbox[c,c]{72pt}{
\begin{tikzpicture}
\tikzstyle{every node}=[draw,circle,fill=black,minimum size=2pt,inner sep=0pt];
\node at (0,0) [draw,circle,fill=white,minimum size=4pt,inner sep=0pt]{};
\node (1) at (0,0) {};
\node (2) at (0.4,0) {};
\node (3) at (0.8,0) {};
\node (4) at (1.2,0) {};
\node (5) at (1.6,0) {};
\node (6) at (2.0,0.2) {};
\node (7) at (2.0,-0.2) {};
\node (8) at (2.4,-0.2) {};
\draw (1)--(2)--(3)--(4)--(5)--(6);
\draw (5)--(7)--(8);
\end{tikzpicture}
} &
$240$ $[126]$ & $(2160_1,1_{2160})$
\\
\rule{0pt}{3ex}
\parbox[c,c]{72pt}{
\begin{tikzpicture}
\tikzstyle{every node}=[draw,circle,fill=black,minimum size=2pt,inner sep=0pt];
\node at (0,0) [draw,circle,fill=white,minimum size=4pt,inner sep=0pt]{};
\node (1) at (0,0) {};
\node (2) at (0.4,0) {};
\node (3) at (0.8,0.2) {};
\node (4) at (1.2,0.2) {};
\node (5) at (0.8,-0.2) {};
\node (6) at (1.2,-0.2) {};
\draw (1)--(2)--(3)--(4);
\draw (2)--(5)--(6);
\end{tikzpicture}
} &
$72$ $[30]$ & $(270_2,54_{10})$
\\
\rule{0pt}{3ex}
\parbox[c,c]{72pt}{
\begin{tikzpicture}
\tikzstyle{every node}=[draw,circle,fill=black,minimum size=2pt,inner sep=0pt];
\node at (0,0) [draw,circle,fill=white,minimum size=4pt,inner sep=0pt]{};
\node (1) at (0,0) {};
\node (2) at (0.4,0) {};
\node (3) at (0.8,0) {};
\node (4) at (1.2,0.2) {};
\node (5) at (1.2,-0.2) {};
\node (6) at (1.6,-0.2) {};
\node (7) at (2.0,-0.2) {};
\draw (1)--(2)--(3)--(4);
\draw (3)--(5)--(6)--(7);
\end{tikzpicture}
} &
$126$ $[60]$ & $(756_2,56_{27})$
\\
\rule{0pt}{3ex}
\parbox[c,c]{72pt}{
\begin{tikzpicture}
\tikzstyle{every node}=[draw,circle,fill=black,minimum size=2pt,inner sep=0pt];
\node at (0,0) [draw,circle,fill=white,minimum size=4pt,inner sep=0pt]{};
\node (1) at (0,0) {};
\node (2) at (0.4,0) {};
\node (3) at (0.8,0.2) {};
\node (4) at (1.2,0.2) {};
\node (5) at (0.8,-0.2) {};
\node (6) at (1.2,-0.2) {};
\node (7) at (1.6,-0.2) {};
\draw (1)--(2)--(3)--(4);
\draw (2)--(5)--(6)--(7);
\end{tikzpicture}
} &
$576$ $[105]$ & $(7560_3,2268_{10})$
\\
\rule{0pt}{3ex}
\parbox[c,c]{72pt}{
\begin{tikzpicture}
\tikzstyle{every node}=[draw,circle,fill=black,minimum size=2pt,inner sep=0pt];
\node at (0,0) [draw,circle,fill=white,minimum size=4pt,inner sep=0pt]{};
\node (1) at (0,0) {};
\node (2) at (0.4,0) {};
\node (3) at (0.8,0) {};
\node (4) at (1.2,0.2) {};
\node (5) at (1.2,-0.2) {};
\node (6) at (1.6,-0.2) {};
\node (7) at (2.0,-0.2) {};
\node (8) at (2.4,-0.2) {};
\draw (1)--(2)--(3)--(4);
\draw (3)--(5)--(6)--(7)--(8);
\end{tikzpicture}
} &
$2160$ $[280]$& $(60480_3,6720_{27})$
\\
\rule{0pt}{3ex}
\parbox[c,c]{72pt}{
\begin{tikzpicture}
\tikzstyle{every node}=[draw,circle,fill=black,minimum size=2pt,inner sep=0pt];
\node at (0,0) [draw,circle,fill=white,minimum size=4pt,inner sep=0pt]{};
\node (1) at (0,0) {};
\node (2) at (0.4,0) {};
\node (3) at (0.8,0.2) {};
\node (4) at (1.2,0.2) {};
\node (5) at (0.8,-0.2) {};
\node (6) at (1.2,-0.2) {};
\node (7) at (1.6,-0.2) {};
\node (8) at (2.0,-0.2) {};
\draw (1)--(2)--(3)--(4);
\draw (2)--(5)--(6)--(7)--(8);
\end{tikzpicture}
} &
$17280$ $[280]$& $(604800_4,241920_{10})$
\\
\end{tabular}
\caption{
Enumerative description of finite figures related to Definition \ref{actions}.
The number of secants through each point on $C$ is included in brackets.
The Pascal lines form a configuration, $\Pi$.
The notation $(p_\alpha,q_\beta)$ denotes a configuration of $p$ points and $q$ lines with $\alpha$ lines through each point and $\beta$ points on each line; the constraint $p\alpha=q\beta$ corresponds to the total number of point-line incidences.
}
\label{main}
\end{center}
\end{table}

\section{Gosset-Elte figures inscribed in a conic}\label{Gosset-Elte}
A combinatorial polytope in a projective space, is a point-line figure whose incidences correspond to the vertex-edge incidences of the polytope.
Concerning the octahedron and its higher dimensional counterparts, namely the cross-polytopes, we use the following terminology.
\begin{definition}\label{central}
A combinatorial cross-polytope with the property that the lines joining opposing vertices (its axes) are all concurrent at a point, will be called a cross-polytope with a centre.
\end{definition}
This allows to describe Figure \ref{skp-illustration2} as an octahedron with a centre inscribed in a conic.
However, the edges are not marked, rather, the lines correspond to the axes, so it is a simpler skeleton from which the octahedron can be recovered.
In the same way, Figures \ref{conic} and \ref{sixteen}, are skeletons of a four-dimensional rectified simplex, and a five-dimensional demicube.

Indeed, the planar figures corresponding to Definition \ref{actions} for values of $(i,j,k)$ satisfying $i+j+k+1\ge ijk-1$, can each be described in terms of a corresponding figure from the Gosset-Elte family of uniform polytopes and tessellations.
Coxeter \cite{cox-vf,rp} introduced this family, giving a unified construction from the reflection group associated with Figure \ref{cdd}, and its members are subsequently identified by the corresponding (Coxeter) symbol $k_{i,j}$. 
Gosset and Elte separately discovered between them the exceptional cases, but the family also includes the infinite sequences of rectified simplexes corresponding to the case $k=0$, the cross-polytopes when $i=j=1$, and the demihypercubes when $i=k=1$.

The connection between Definition \ref{actions} and the combinatorial Gosset-Elte figures is based on the idea of a {\it determining set}.
Under the two constraints that (i), vertices correspond to points on a conic, and (ii), each octahedral cell has a centre, the remainder of the figure is determined uniquely from a freely chosen sub-figure.
This sub-figure, which does not contain any whole octahedral cell, will be described first, and then related to the full figure.

\begin{definition}\label{ds}
The figure associated with initial data arrays (\ref{Xdef}):
View the entries of arrays (\ref{Xdef}) as parameters corresponding to points on a conic, and add lines joining pairs of points whose entries, (i) are from distinct arrays, (ii) both belong to the first array, and, (iii) both belong to a common row or column of the second array. 
\end{definition}

The combinatorial polytope described in Definition \ref{ds} can be called a {\it compound simplex}, its associated Coxeter-graph is obtained from Figure \ref{cdd} by deleting the node labelled $r_0$.

\begin{prop}\label{description}
For generic values of the parameters, Definition \ref{ds} is a determining sub-figure of a combinatorial $k_{i,j}$-polytope or tessellation inscribed in a conic, constrained by the condition that all cross-polytope (or $k_{1,1}$) facets have a centre (Definition \ref{central}).
Birational actions generated by (\ref{rational_actions}) determine the image of this sub-figure under combinatorial symmetries of the $k_{i,j}$-figure.
\end{prop}
\begin{proof}
The correspondence associating variables to vertices, and equations to octahedral cells of the $k_{i,j}$-polytopes and tessellations, was made previously in \cite{SKP}.
In outline, the discrete domain, like the corresponding $k_{i,j}$-figure, is a realisation of an incidence structure expressible in terms of cosets of the underlying Coxeter group.
The proposition is therefore implied directly by Lemma \ref{gg}, because it is clear that imposing that all octahedral cells have a centre, is equivalent to imposing that the cross-polytope facets do.
The set of initial values for the discrete system corresponds to the determining-set of the figure. The consistency of the initial-value-problem was established previously by reducing it to verifying the relations satisfied by the associated actions (\ref{rational_actions}).
That the figure in its entirety is determined, rather than some part of it, follows from the fact that the combinatorial symmetries are transitive on vertices and edges of the $k_{i,j}$ figures.
\end{proof}

Because the vertices of the combinatorial $k_{i,j}$-figure are confined to a conic, the centres of all cross-polytope facets belonging to the same $k_{1,2}$ or $k_{2,1}$-facet, are collinear.
In the case of the rectified four-simplex ($0_{1,2}$-polytope) this was established as part of Proposition \ref{direct}, and the more general assertion follows by repeated application of this case.
\begin{remark}\label{Pi}
The cross-polytope centres (Proposition \ref{description}) are therefore points of a planar point-line configuration ($\Pi$ in Table \ref{main}), which is a projective realisation of the incidence structure given by associating points with $k_{1,1}$-facets, and lines with the $k_{2,1}$ and $k_{1,2}$-facets, of the $k_{i,j}$-figure.
\end{remark}

\section{Ambient solution}\label{solutions}
An application of this geometric view of Coble's group (Proposition \ref{description}, Remark \ref{Pi}), is that the special situation when {\it all} cross-polytope centres align (the configuration $\Pi$ in Table \ref{main} collapses to a line) can be identified, that suggests the birational group should linearise in terms of the associated geometric group on the conic.

To formulate this, recall first the participating group.
Denote the conic and the single Pascal line by $C$ and $\pi$ respectively.
For any $a,b,e\in C\setminus \pi$, there exists a unique point $c\in C\setminus \pi$ such that the lines determined by the pairs $\{a,b\}$ and $\{c,e\}$, intersect $\pi$ at the same point.
The resulting mapping $(a,b)\mapsto c$ turns $C\setminus \pi$ into a group with identity $e$.
In terms of corresponding parameters $(x,y,z)$ for $(a,b,c)$, use the notation $z=x*y$ for the product, and $x^{-1}$ for the inverse of $x$.
This product is equivalent to either addition or multiplication in the field, depending on how $\pi$ meets $C$.
\begin{prop}\label{solution}
The cross-polytope centres described in Proposition \ref{description} lie on a common line $\pi$ if, and only if, elements of the array on the right in (\ref{Xdef}) are such that
\begin{equation}\label{fortress}
y_{mn}=y_{00}^{-1}*y_{m0}*y_{0n}, \quad m\in\{1,\ldots, i\},n\in\{1,\ldots, j\}. 
\end{equation}
This reduces the birational actions (\ref{rational_actions}) to linear ones for the remaining variables:
\begin{equation}\label{linear-actions}
\begin{array}{rll}
t_m: & y_{(m-1)0} \leftrightarrow y_{m0}, & m \in \{2,\ldots,i\},\\
t_1: & y_{00}\leftrightarrow y_{10}, \ y_{0n} \rightarrow y_{00}^{-1}*y_{10}*y_{0n}, & n \in \{1,\ldots,j\},\\
s_n: & y_{0(n-1)} \leftrightarrow y_{0n}, & n\in\{2,\ldots,j\},\\
s_1: & y_{00}\leftrightarrow y_{01}, \ y_{m0} \rightarrow y_{00}^{-1}*y_{01}*y_{m0}, & m \in \{1,\ldots,i\},\\
r_n: & y_{n-1} \leftrightarrow y_{n}, & n\in\{1,\ldots,k\},\\
r_0: & y_0\leftrightarrow y_{00}.
\end{array}
\end{equation}
\end{prop}
\begin{proof}
The intersection of lines determined by the two pairs $\{y_{mn},y_{\alpha\beta}\}$ and $\{y_{m\beta},y_{\alpha n}\}$ taken from the array on the right in (\ref{Xdef}), where $m\neq \alpha$ and $n\neq \beta$, is the centre of an octahedral cell.
The condition for these centres to be on $\pi$ can be expressed in terms of the geometric group on $C\setminus \pi$, as
\begin{equation}\label{reinforced}
y_{mn}*y_{\alpha\beta}=y_{m\beta}*y_{\alpha n}, \quad m,\alpha\in \{0,\ldots, i\},\ n,\beta\in \{0,\ldots, j\}.
\end{equation}
The conditions (\ref{fortress}) are clearly a subset of the conditions (\ref{reinforced}).
That (\ref{fortress}) implies (\ref{reinforced}) is established by substitution, relying on associativity of the group and therefore Pascal's theorem in geometric terms.

To verify that all octahedral centres are on $\pi$, it is sufficient to show that the actions (\ref{rational_actions}) preserve the condition (\ref{fortress}).
For the actions $t_1,\ldots,t_i,s_1,\ldots,s_j$ this follows from the established equivalence between (\ref{fortress}) and (\ref{reinforced}).
For the action $r_0$ it follows from the concurrence of lines determined by the three pairs of points (\ref{sixvars}), which was the content of Lemma \ref{gg}.
For the remaining actions $r_1,\ldots,r_k$ it is trivial.
\end{proof}
In the case $k_{1,2}$, corresponding to the sequence in Table \ref{sequence}, there is only a single Pascal line, so there is no loss of generality if (\ref{fortress}) is assumed.
The subgroup associated with the $1_{5,2}$-tessellation is known to be linearised by substitution of elliptic functions, and Proposition \ref{solution} shows that rational functions are sufficient for the $5_{2,1}$-tessellation.
Nevertheless, it determines an $\tilde{E}_8$ action on $\mathbb{P}^1\times \mathbb{P}^1$ by equivariant extension to the case $5_{3,1}$.
It would therefore be interesting to know where this fits in relation to the QRT maps \cite{qrt1} and classification of discrete Painlev\'e equations \cite{sc}.

\remark{
The linear representation (\ref{linear-actions}) of the Coxeter group encoded in Figure \ref{cdd} can of course be realised in any abelian group, and in particular when the conic is replaced with a cubic curve.
However, in the general case, when the representation of the Coxeter group is birational (Section \ref{Coble}),
the algebraic view is along the lines described by M\"obius, relying on closure in the non-commutative M\"obius group, which is incompatible with the generalisation to a cubic curve.
In other words, if the combinatorial Gosset-Elte figures, with the same constraint that octahedral cells have a centre, were to be inscribed in a cubic curve, it would correspond to the linear group (\ref{linear-actions}) when the octahedral centres are also on the curve, but a group that is not in general birational when the positions of the centres are unconstrained.
}

\section{Associated circle patterns}\label{Clifford}
If the projective plane over $\mathbb{C}$ is viewed as a four-dimensional real space, then a conic in the plane is seen as a quadric surface in the space.
Three points on the conic correspond to points on the surface that are either on a line within the surface, or determine a plane that cuts it.
In this way, the surface constitutes an inversive plane; each curve determined by a set of three points, corresponds to a circle.
\begin{remark}\label{four-points}
The ambient notion of general-position for sets of points in an inversive plane, has the requirement that no four points be on the same circle.
The corresponding requirement is absent in the projective setting.
\end{remark}

In this model of the inversive plane, equation (\ref{skp}) is an analytic form of a geometric constraint established in \cite{ks1,ks2}, that the four circles determined by sets of points
\begin{equation}\label{faces1} 
\{x_{12},x_{23},x_{13}\}, 
\{x_{23},x_{34},x_{24}\}, 
\{x_{13},x_{34},x_{14}\}, 
\{x_{12},x_{24},x_{14}\},
\end{equation}
are concurrent at a point.
Or, equivalently, due to Clifford's four-circle theorem, that the circles determined by the sets
\begin{equation} \label{faces2}
\{x_{12},x_{13},x_{14}\}, 
\{x_{12},x_{23},x_{24}\}, 
\{x_{13},x_{23},x_{34}\}, 
\{x_{14},x_{24},x_{34}\},
\end{equation}
are.
The resulting figure of eight circles and eight points, is the $C_4$ pattern appearing first in Clifford's chain of theorems.
By symmetry, the condition on the participating points can be reduced to the equivalent one, that any three of the circles meet at a point.
The sets of points (\ref{faces1}) and (\ref{faces2}) correspond to the eight faces of the octahedron.
In summary.
\begin{coro}[of Proposition \ref{description}] \label{circles}
View the entries of arrays (\ref{Xdef}) as points in the inversive plane.
In general position, these points determine a combinatorial $k_{i,j}$-polytope or tessellation, constrained by the condition that the vertices of each octahedral cell correspond to points of a $C_4$ circle pattern (see above).
When the inversive plane is identified with $\mathbb{C}\cup\{\infty\}$, the actions (\ref{rational_actions}) give the image of the determining-set under combinatorial symmetries of the $k_{i,j}$-figure.
\end{coro}
The statement of existence and movability of the circle patterns in the case $k=0$ was given originally in relation to the integrability of equation (\ref{skp}) in \cite{ks1}.
Longuet-Higgins \cite{LH} actually related Clifford's chain to the $1_{1,j}$-polytopes (demihypercubes), and added new examples to previously known configurations of points and hyperspheres in connection with the remaining (finite) list of $1_{i,j}$-polytopes.
The patterns here are planar, so they have a different nature, but they exist uniformly in relation to all of the $k_{i,j}$-polytopes and tessellations.

The mismatch between patterns corresponding to our and the Longuet-Higgins Coxeter symbol, is seen already in the case of the $C_4$ pattern described above.
The two additional points are on the same footing as the original ones, and the eight taken together correspond to vertices of the four-dimensional cross-polytope, whereas the original constraint corresponds to the octahedron.
This feature is related to Remark \ref{four-points}, however, a generalisation to higher dimension without this feature as been obtained \cite{CSC}, and it would be interesting to know if those patterns have a projective formulation.

Circle patterns for Riccati solutions of discrete Painlev\'e equations have been related to discrete analogues of holomorphic functions \cite{Aga,AgaBob}. 
The circle patterns here are instead associated with the general solution, but the hypergeometric solutions of the elliptic Painlev\'e equation \cite{10E9} are also framed geometrically in terms of a conic.

\section{Desargues maps}\label{Desargues}
In terms of the Coxeter symbol, the Desargues maps of Doliwa \cite{DDM,Daff} in the fundamental region, are a combinatorial $0_{i,j}$-polytope satisfying the condition that the vertices of each $0_{i,0}$-facet are collinear.
It means some edges coalesce resulting in a point-line configuration; in summary we emphasise the following.
\begin{remark}\label{Doliwa}
In the fundamental region the Desargues map is a point-line realisation of the incidence structure defined by vertices and $0_{i,0}$-facets of the $0_{i,j}$-polytope.
\end{remark}
This is a 
$
(p_{j+1},q_{i+2}),\ p=(i+j+2)!/[(i+1)!(j+1)!], 
$
configuration, the first few cases are given in Table \ref{configs}.
Although they are both simplexes, the $0_{0,j}$-facets are distinguished from the $0_{i,0}$ ones, which breaks the antipodal symmetry of the $0_{i,i}$-polytopes, or equivalently the Coxeter graph automorphism.
There is a direct relationship between the Desargues maps whose image is in $\mathbb{A}^M$, $M>1$, over a field, and the form of Coble's group in Definition \ref{actions}, which leads to the following.
\begin{table}[t!]
\begin{center}
\begin{tabular}{lllll}
$(6_2,4_3)$ & $(10_3,10_3)$ & $(15_4,20_3)$ & $(21_5,35_3)$ & $(28_6,56_3)$ \\
$(10_2,5_4)$ & $(20_3,15_4)$ & $(35_4,35_4)$ & $(56_5,70_4)$ & $(84_6,126_4)$ \\
$(15_2,6_5)$ & $(35_3,21_5)$ & $(70_4,56_5)$ & $(126_5,126_5)$ & $(210_6,252_5)$ \\
$(21_2,7_6)$ & $(56_3,28_6)$ & $(126_4,84_6)$ & $(252_5,210_6)$ & $(462_6,462_6)$ \\
$(28_2,8_7)$ & $(84_3,36_7)$ & $(210_4,120_7)$ & $(462_5,330_7)$ & $(924_6,792_7)$
\end{tabular}
\end{center}
\caption{Enumerative description of the first few of Doliwa's configurations associated with Coxeter symbol $0_{i,j}$. The balanced $(10_3)$, corresponding to Coxeter symbol $0_{1,2}$, is Desargues' configuration.}
\label{configs}
\end{table}

\begin{prop}\label{dd}
Freely choose a point in $\mathbb{A}^M$, $j+1$ lines through it, and $i+1$ additional points on each line.
In general position, this is a determining-set for Doliwa's configuration associated with Coxeter symbol $0_{i,j}$ (Remark \ref{Doliwa}).

Denote the first freely chosen point by $y_0$ and, for each $n\in\{0,\ldots,j\}$, denote the $i+1$ additional points on the $n^{th}$ line passing through $y_0$, by $y_{0n},\ldots,y_{in}$.
Then generators (\ref{rational_actions}) in the case $k=0$, acting diagonally on $\mathbb{A}^M$, give the image of this determining-set under combinatorial symmetries of the configuration.
\end{prop}
\begin{proof}
The connection between this geometric setting and equation (\ref{skp}) can be formulated as follows.
Suppose lines $l_1,\ldots,l_4 \subset \mathbb{A}^M$ intersect pairwise,
\begin{equation}
\begin{split}
&l_1 \cap l_2 = x_{12}, \quad
l_1 \cap l_3 = x_{13}, \quad
l_2 \cap l_3 = x_{23}, \\
&l_1 \cap l_4 = x_{14}, \quad
l_2 \cap l_4 = x_{24}, \quad
l_3 \cap l_4 = x_{34}, 
\end{split}
\end{equation}
for some set of points $x_{12},\ldots,x_{34}\in \mathbb{A}^M$, forming the configuration of a complete quadrilateral $(6_2,4_3)$.
Then the points satisfy the analytic condition
\begin{equation}
\frac{(x^{{(\alpha)}}_{12}-x^{{(\alpha)}}_{24})(x^{{(\alpha)}}_{13}-x^{{(\alpha)}}_{23})(x^{{(\alpha)}}_{14}-x^{{(\alpha)}}_{34})}{(x^{{(\alpha)}}_{12}-x^{{(\alpha)}}_{23})(x^{{(\alpha)}}_{14}-x^{{(\alpha)}}_{24})(x^{{(\alpha)}}_{13}-x^{{(\alpha)}}_{34})}=1, \quad \alpha \in \{1,\ldots,M\},\label{component}
\end{equation}
in which the affine coordinate notation $x_{mn} = (x_{mn}^{(1)},\ldots,x_{mn}^{(M)})$ has been used.
Because the condition (\ref{component}) determines one point from the others, it follows that the diagonal actions described in the proposition preserve the collinearity assumed for the points of the determining-set; it is trivial for all actions except $r_0$, and for this one it is sufficient to consider the case $i=j=1$, $M=2$.
In a projective setting, the multi-ratio condition emerges by projection of the quadrilateral from a point onto a line,
equivalent to the above for the affine case, and this has been described in earlier works \cite{Bog,DDM}.
The affine setting used in this proposition is practical, and not really a restriction: it differs by the implied meaning of {\it general-position} for the determining sub-figure. 

Observe now (separately from the above) that the $0_{i,0}$ sub-graph of generators (\ref{rational_actions}), i.e., $r_0,t_1,\ldots,t_i$, freely permutes entries $y_0,y_{00},y_{10},\ldots,y_{i0}$ of arrays (\ref{Xdef}), so in accordance with the correspondence described in Proposition \ref{description}, these entries correspond to vertices of a $0_{i,0}$-facet.
Also notice, that all rows of the array on the right in (\ref{Xdef}) are on the same footing, because the rows are themselves freely permuted by the actions $s_1,\ldots,s_j$.
Therefore the entries of rows of the array on the right in (\ref{Xdef}) correspond to vertices of $j+1$ $0_{i,0}$-facets whose common vertex corresponds to $y_0$.

The second paragraph of the Proposition follows by combining these two observations.
The first shows that the diagonal actions described determine a point-line figure in $\mathbb{A}^M$, and the second that it coincides with Doliwa's $0_{i,j}$ configuration.

It is clear the general configuration is determined by the group actions, because, first, all that is imposed is the linearity and multi-ratio constraints on each octahedral cell, which are both necessary.
Second, that the whole configuration is obtained, and not some part of it, follows from the fact that the group of combinatorial symmetries of the $0_{i,j}$ polytope acts transitively on vertices.
This establishes the first paragraph of the Proposition.
\end{proof}
\begin{table}[t!]
\begin{center}
\begin{tabular}{lllll}
\!\! -&\!\!$(5_1,1_5)$&\!\!$(15_2,6_5)$&\!\!$(35_3,21_5)$&\!\!$(70_4,56_5)$\\
\!\!$(5_1,1_5)$&\!\!$(30_2,12_5)$&\!\!$(105_3,63_5)$&\!\!$(280_4,224_5)$&\!\!$(630_5,630_5)$\\
\!\!$(15_2,6_5)$&\!\!$(105_3,63_5)$&\!\!$(420_4,336_5)$&\!\!$(1260_5,1260_5)$&\!\!$(3150_6,3780_5)$\\
\!\!$(35_3,21_5)$&\!\!$(280_4,224_5)$&\!\!$(1260_5,1260_5)$&\!\!$(4200_6,5040_5)$&\!\!$(11550_7,16170_5)$\\
\!\!$(70_4,56_5)$&\!\!$(630_5,630_5)$&\!\!$(3150_6,3780_5)$&\!\!$(11550_7,16170_5)$&\!\!$(34650_8,55440_5)$
\end{tabular}
\end{center}
\caption{Enumerative description of the first few configurations associated with Coxeter symbol $0_{i,j}$ described in Remark \ref{Pi}. Entry corresponding to Coxeter symbol $0_{2,2}$ is Schl\"afli's double-six.}
\label{configs2}
\end{table}
This proposition shows that the Desargues maps whose image is in $\mathbb{A}^M$ over a field, decompose into $M$ simpler independent systems, each equivalent to case $k=0$ of Coble's group.
The corresponding result in a projective space over a skew-field is likely to hold, because the consistency for Doliwa's configurations is encoded in Desargues' theorem, and not Pascal's \cite{DDM,ks2}. The suitable notion of general-position for the inital data, will be weaker in the projective case. 
The circle patterns corresponding to the case $k=0$ have also been generalised in this direction \cite{CSC}.
Given the view established by Lemma \ref{gg}, it would therefore be interesting to know if this form of Coble's group can extend to the skew-field setting in the case $k>0$.
In this regard, a basic question to shed some light, is to know the geometric origin of the consistency property underlying the movability of the point-line configurations $\Pi$ of Remark \ref{Pi}.
The first few cases associated to the Coxeter symbol $0_{i,j}$ are displayed Table \ref{configs2} for comparison with Table \ref{configs}.

The connection between the Desargues maps and the polytope inscribed in a conic, is not a geometric one.
It is possible to see the Desargues maps as a special case of the circle patterns \cite{ks2} in which all circles become lines, but they then appear in a restricted two-dimensional setting.
However, Proposition \ref{dd} is equivalent to a symmetry of the actions (\ref{rational_actions}) present when $k=0$.
A similar kind of symmetry is also present in the case $k=1$; it is inferred from the integration that recovers Coble's group in its original form \cite{SKP}.
It is therefore an interesting question to understand the geometric origin of these symmetries in the context of Lemma \ref{gg}.

\vspace{10pt}
\noindent {\bf\large Acknowledgements}\\
\noindent I am very grateful to Pavlos Kassotakis for our discussions on this topic.

\appendix

\section{Proof of Lemma \ref{gg}}\label{roots}
There is an equivalence between identifying a polynomial of degree $n$ with its set of factors, and the identification of a hyperplane in $\mathbb{P}^n$ by the set of points where it intersects the {\it rational normal curve}, 
\[ (a^n:a^{n-1}b:\ldots:ab^{n-1}:b^n), \quad (a:b)\in\mathbb{P}^1, \]
the order of the contact with the curve corresponds to multiplicity of the factor.
This can be used in the case $n=2$, when the curve is a conic, to verify Lemma \ref{gg}.

In the polynomial context, the main observation is that:
\begin{lemma}\label{gi}
Equation (\ref{skp}) is the condition for linear dependence of the three quadratic polynomials, 
\begin{equation}\label{tp}
(x-x_{12})(x-x_{34}),\ 
(x-x_{13})(x-x_{24}),\ 
(x-x_{14})(x-x_{23}).
\end{equation}
\end{lemma}
\begin{proof}
The equation (\ref{skp}) can be written differently as
\begin{equation}\label{skp2}
\left|
\begin{array}{ccc}
x_{12}x_{34} & x_{12}+x_{34} & 1\\
x_{13}x_{24} & x_{13}+x_{24} & 1\\
x_{14}x_{23} & x_{14}+x_{23} & 1
\end{array}
\right| = 0,
\end{equation}
which is the condition on polynomial coefficients for the linear dependence.
\end{proof}
Take the representation of $\mathbb{P}^2$ where {\it lines} and {\it points}, respectively, are represented by the {\it one} and {\it two} dimensional subspaces in the vector space of polynomials whose degree is less than or equal to two.
This is dual to taking polynomial coefficients as projective coordinates.
In this representation, the subspaces of polynomials defined by a shared root,
\begin{equation}\label{con}
\langle
ax-b,x(ax-b)
\rangle, \quad (a:b)\in\mathbb{P}^1,
\end{equation}
represent the points of a conic, and the subspaces of polynomials with a double root,
\begin{equation}\label{tan}
\langle
(ax-b)^2
\rangle,
\end{equation}
represent its corresponding tangent lines.
Distinct elements $(a:b),(c:d)\in\mathbb{P}^1$ determine a secant,
\begin{equation}\label{sec}
\langle
ax-b,x(ax-b)
\rangle
\cap
\langle
cx-d,x(cx-d)
\rangle
=
\langle
(ax-b)(cx-d)
\rangle.
\end{equation}
The linear dependence of polynomials (\ref{tp}) is concurrence of corresponding secant or tangent lines, and therefore Lemmas \ref{gg} and \ref{gi} are equivalent.

One reason to verify Lemma \ref{gg} in this way, is due to the related form of Coble's group established in \cite{ib,ay}.
That form involves the case $n=3$ above, of concurrent planes passing through a twisted cubic curve, and provides a view of the group complementary to the one here, that will be explained in a forthcoming paper.

\bibliographystyle{unsrt}
\bibliography{references}
\end{document}

%% file: skp-pic2.tex
\begin{tikzpicture}[line cap=round,line join=round,>=triangle 45,x=1.0cm,y=1.0cm,scale=0.7]
\clip(-2.198675260816108,-0.3090233135157237) rectangle (5.086795530131022,5.632403728067066);
\draw [rotate around={-178.84645628602107:(1.4894157407604034,3.003671319090707)}] (1.4894157407604034,3.003671319090707) ellipse (2.4579728352218475cm and 1.2461351422962144cm);
\draw [domain=-2.198675260816108:5.086795530131022] plot(\x,{(--5.330496500877649--1.9382597631593363*\x)/3.084127423201307});
\draw [domain=-2.198675260816108:5.086795530131022] plot(\x,{(--15.538950967737803-1.1394680579201548*\x)/4.163241033055564});
\draw [domain=-2.198675260816108:5.086795530131022] plot(\x,{(-4.948149101844914--1.059392261900026*\x)/-0.8304954459702931});
\begin{scriptsize}
\draw [fill=black] (0.2758725767986928,1.9017402368406635) circle (1.2pt);
\draw (0.22,2.19) node {${14}$};
\draw [fill=black] (-0.32,3.82) circle (1.2pt);
\draw (-0.35,4.15) node {${24}$};
\draw [fill=black] (3.36,3.84) circle (1.2pt);
\draw (3.35,4.2) node {${23}$};
\draw [fill=black] (3.8432410330555635,2.680531942079845) circle (1.2pt);
\draw (4.3,2.8) node {${13}$};
\draw [fill=black] (3.0518881712802717,2.0650352718101694) circle (1.2pt);
\draw (3.15,2.46) node {${34}$};
\draw [fill=black] (1.3424913796183775,4.2455670765181255) circle (1.2pt);
\draw (1.56,4.6) node {${12}$};
\end{scriptsize}
\end{tikzpicture}

%% file: octahedron.tex
\begin{tikzpicture}[line join=bevel,z=-5.5,scale=2.0]
\coordinate (Y3) at (0,0,-1);
\coordinate (Y2) at (-1,0,0);
\coordinate (X3) at (0,0,1);
\coordinate (X2) at (1,0,0);
\coordinate (X1) at (0,1,0);
\coordinate (Y1) at (0,-1,0);

\node at (X1) [label={[label distance=0.0cm]above:${12}$}]{};
\node at (X2) [label={[label distance=0.0cm]right:${13}$}]{};
\node at (Y1) [label={[label distance=0.0cm]below:${34}$}]{};
\node at (Y2) [label={[label distance=0.0cm]left:${24}$}]{};
\node at (Y3) [label={[label distance=0.0cm]30:${23}$}]{};


\draw (Y3) -- (Y2) -- (X1) -- cycle;
\draw (X2) -- (Y3) -- (X1) -- cycle;
\draw (Y3) -- (Y2) -- (Y1) -- cycle;
\draw (X2) -- (Y3) -- (Y1) -- cycle;

\draw [fill opacity=0.7,fill=white!90!black] (Y2) -- (X3) -- (X1) -- cycle;
\draw [fill opacity=0.7,fill=white!80!black] (X3) -- (X2) -- (X1) -- cycle;
\draw [fill opacity=0.7,fill=white!70!black] (Y2) -- (X3) -- (Y1) -- cycle;
\draw [fill opacity=0.7,fill=white!50!black] (X3) -- (X2) -- (Y1) -- cycle;

\node at (X3) [label={[label distance=0.0cm]210:${14}$}]{};

\end{tikzpicture}

%% file: conic_theorem.tex
\begin{center}
\begin{tikzpicture}[line cap=round,line join=round,>=triangle 45,x=1.0cm,y=1.0cm,scale=0.145]
\clip(-35.59888079033262,-30.232241679278378) rectangle (46.78096107862958,33.8935135298561);
\draw [rotate around={-133.38297790197979:(6.1357134844477486,-0.30966575866374424)}] (6.1357134844477486,-0.30966575866374424) ellipse (17.509963157450795cm and 9.99119103801839cm);
\draw [color=cqcqcq,domain=-35.59888079033262:46.78096107862958] plot(\x,{(-78.82946684184145--22.976250103022032*\x)/25.183008587493124});
\draw [color=cqcqcq,domain=-35.59888079033262:46.78096107862958] plot(\x,{(--128.46110013262887-28.250385712223654*\x)/-10.0187591686387});
\draw [color=cqcqcq,domain=-35.59888079033262:46.78096107862958] plot(\x,{(-66.60079804345239--11.578879154390894*\x)/-5.157806967309027});
\draw [color=cqcqcq,domain=-35.59888079033262:46.78096107862958] plot(\x,{(-223.4018084705489--11.737843848648456*\x)/-1.3275754186954742});
\draw [color=cqcqcq,domain=-35.59888079033262:46.78096107862958] plot(\x,{(-253.27278635200958--17.531667146856694*\x)/-6.650461854672489});
\draw [color=cqcqcq,domain=-35.59888079033262:46.78096107862958] plot(\x,{(-34.97386501151021-7.152091984491721*\x)/10.741807104331148});
\draw [color=cqcqcq,domain=-35.59888079033262:46.78096107862958] plot(\x,{(--278.2341908159551-15.92839991180345*\x)/-19.589226573446613});
\draw [color=cqcqcq,domain=-35.59888079033262:46.78096107862958] plot(\x,{(-191.7552463728867--6.028724907937087*\x)/23.590578456053173});
\draw [color=cqcqcq,domain=-35.59888079033262:46.78096107862958] plot(\x,{(-55.03514491049369--4.729150323050161*\x)/0.6565228987825513});
\draw [color=cqcqcq,domain=-35.59888079033262:46.78096107862958] plot(\x,{(--59.563420054100504--10.526395624808982*\x)/2.0390703024647197});
\draw [color=cqcqcq,domain=-35.59888079033262:46.78096107862958] plot(\x,{(-171.22324206698414--7.99681839879012*\x)/-17.230125074041457});
\draw [color=cqcqcq,domain=-35.59888079033262:46.78096107862958] plot(\x,{(-51.47519356765841-0.6561566078271852*\x)/-4.27052048397473});
\draw [color=cqcqcq,domain=-35.59888079033262:46.78096107862958] plot(\x,{(--153.17385239496596-13.206499745226612*\x)/14.310119523906032});
\draw [color=cqcqcq,domain=-35.59888079033262:46.78096107862958] plot(\x,{(--82.55313902047814--15.216389282238854*\x)/-4.8822871302772635});
\draw [color=cqcqcq,domain=-35.59888079033262:46.78096107862958] plot(\x,{(-73.05181583729615--1.2402985595989318*\x)/-3.97237150213771});
\draw [domain=-35.59888079033262:46.78096107862958] plot(\x,{(--251.48020971012843-39.54525108267488*\x)/33.56762582776311});
\begin{scriptsize}
\draw [fill=black] (10.239019887098344,-10.073191226032082) circle (5.0pt);
\draw (11.218595357929479,-8.76240918300474) node {12};
\draw [fill=black] (9.330354628786736,13.487200828761699) circle (5.0pt);
\draw (10.272787758974689,14.788200030969426) node {13};
\draw [fill=black] (-7.609761972593947,-10.073191226032082) circle (5.0pt);
\draw (-6.657168262316051,-8.76240918300474) node {14};
\draw [fill=black] (-0.6884045398519641,-14.763184883461953) circle (5.0pt);
\draw (0.2472272100539159,-13.39686641788319) node {24};
\draw [fill=black] (17.573246614899176,12.903058876989952) circle (5.0pt);
\draw (18.501313869881365,14.220715471596554) node {23};
\draw [fill=black] (5.0812129197893166,1.5056879283588123) circle (5.0pt);
\draw (5.732911283991697,2.871024284139124) node {5};
\draw [fill=black] (1.6706969595531938,9.162033427131616) circle (5.0pt);
\draw (2.611746207440891,10.53206583567289) node {34};
\draw [fill=black] (18.90082203359465,1.1652150283414962) circle (5.0pt);
\draw (19.82544450841807,2.49270124455721) node {15};
\draw [fill=black] (15.980816483459225,-4.044466318094995) circle (5.0pt);
\draw (16.89344095165822,-2.709240549694112) node {25};
\draw [fill=black] (20.980826991429492,-17.225283210523802) circle (5.0pt);
\draw (21.622478946432167,-15.855966175165634) node {4};
\draw [fill=black] (-5.570691670129228,0.4532043987769008) circle (5.0pt);
\draw (-4.670972304510992,1.7360551653933811) node {35};
\draw [fill=black] (10.895542785880895,-5.344040902981921) circle (5.0pt);
\draw (11.596918397511397,-4.033371188230813) node {3};
\draw [fill=black] (13.600875112761466,14.143357436588884) circle (5.0pt);
\draw (14.528921954271244,15.450265350237776) node {45};
\draw [fill=black] (-3.425529044949728,11.527285780325508) circle (5.0pt);
\draw (-2.779357106601412,12.896584833059855) node {2};
\draw [fill=black] (-12.586798836333617,22.31996787215108) circle (5.0pt);
\draw (-11.953690816462874,23.678791461144414) node {1};
\end{scriptsize}
\end{tikzpicture}
\end{center}

%% file: sixteen.tex
\begin{center}
\begin{tikzpicture}[line cap=round,line join=round,>=triangle 45,x=1.0cm,y=1.0cm,scale=0.9]
\clip(-4.712094853111268,-5.802605272107039) rectangle (8.156577630748902,3.599021955179625);
\draw [rotate around={147.4262692515112:(1.5865178729429967,-0.5085611946674883)}] (1.5865178729429967,-0.5085611946674883) ellipse (4.960217582561025cm and 3.1219938259954394cm);
\draw [domain=-4.712094853111268:8.156577630748902] plot(\x,{(-5.988951335236108--1.111779660108572*\x)/1.1319493086218329});
\draw [color=cqcqcq,domain=-4.712094853111268:8.156577630748902] plot(\x,{(-19.71575034124377--1.1831313910930636*\x)/5.939655033120865});
\draw [color=cqcqcq,domain=-4.712094853111268:8.156577630748902] plot(\x,{(-16.958023288945896--3.7048259737905256*\x)/-7.094311325321739});
\draw [color=cqcqcq,domain=-4.712094853111268:8.156577630748902] plot(\x,{(--17.364268854468058-5.204272165403938*\x)/-3.075779054025278});
\draw [color=cqcqcq,domain=-4.712094853111268:8.156577630748902] plot(\x,{(--3.892201343093708--6.371176301002533*\x)/-1.4740192330971345});
\draw [color=cqcqcq,domain=-4.712094853111268:8.156577630748902] plot(\x,{(--17.561488393493565-3.542467241225112*\x)/1.9044563751610788});
\draw [color=cqcqcq,domain=-4.712094853111268:8.156577630748902] plot(\x,{(--21.05224351546152-3.14502386029777*\x)/-4.660872488290096});
\draw [color=cqcqcq,domain=-4.712094853111268:8.156577630748902] plot(\x,{(-10.741989490933975--5.897797812220428*\x)/-1.4578973409027012});
\draw [color=cqcqcq,domain=-4.712094853111268:8.156577630748902] plot(\x,{(--6.072335221908926--3.668588136860572*\x)/-5.433130606526194});
\draw [color=cqcqcq,domain=-4.712094853111268:8.156577630748902] plot(\x,{(--13.568333081968044--5.650891525132896*\x)/-3.414864730341296});
\draw [color=cqcqcq,domain=-4.712094853111268:8.156577630748902] plot(\x,{(-6.626972183853152--7.880101200492752*\x)/0.5603685352821968});
\draw [color=cqcqcq,domain=-4.712094853111268:8.156577630748902] plot(\x,{(-16.018433618408775--3.3560070270577715*\x)/-4.045800704456992});
\draw [color=cqcqcq,domain=-4.712094853111268:8.156577630748902] plot(\x,{(-3.5393409720681817--1.1267973516979155*\x)/-8.021033970080484});
\draw [color=cqcqcq,domain=-4.712094853111268:8.156577630748902] plot(\x,{(-16.16769704663247--1.6850375284669115*\x)/6.714431885172056});
\draw [color=cqcqcq,domain=-4.712094853111268:8.156577630748902] plot(\x,{(-19.245243692471853--2.3604329279662415*\x)/5.742888020287932});
\draw [color=cqcqcq,domain=-4.712094853111268:8.156577630748902] plot(\x,{(-20.130238470252895--3.186960970082357*\x)/1.4525485268981804});
\draw [color=cqcqcq,domain=-4.712094853111268:8.156577630748902] plot(\x,{(-21.174768668801253--2.839106461052004*\x)/4.783468416353424});
\draw [color=cqcqcq,domain=-4.712094853111268:8.156577630748902] plot(\x,{(-10.967956826590811-1.615240552020739*\x)/3.6882229780241467});
\draw [color=cqcqcq,domain=-4.712094853111268:8.156577630748902] plot(\x,{(-20.85570160958727--3.417330257072659*\x)/4.5285306354597274});
\draw [color=cqcqcq,domain=-4.712094853111268:8.156577630748902] plot(\x,{(--23.60388869769396-5.552126674434291*\x)/0.25514083542996513});
\draw [color=cqcqcq,domain=-4.712094853111268:8.156577630748902] plot(\x,{(--8.235267903602207-7.197704343118649*\x)/5.764358726486886});
\draw [color=cqcqcq,domain=-4.712094853111268:8.156577630748902] plot(\x,{(--3.2356394319136452-7.311021453523942*\x)/4.190698346237157});
\draw [color=cqcqcq,domain=-4.712094853111268:8.156577630748902] plot(\x,{(--13.24539991752483--0.5862271107998422*\x)/-3.8504985654903106});
\draw [color=cqcqcq,domain=-4.712094853111268:8.156577630748902] plot(\x,{(--20.252690962794166-1.423432322409337*\x)/-5.499814105221424});
\draw [color=cqcqcq,domain=-4.712094853111268:8.156577630748902] plot(\x,{(--17.64910824543158-0.5077359915937336*\x)/-6.911198898004988});
\draw [color=cqcqcq,domain=-4.712094853111268:8.156577630748902] plot(\x,{(--20.000476617704777-2.906651258528281*\x)/-5.180451164325164});
\draw [color=cqcqcq,domain=-4.712094853111268:8.156577630748902] plot(\x,{(-18.49344935619166--2.9386567239868966*\x)/5.4879502393942365});
\draw [color=cqcqcq,domain=-4.712094853111268:8.156577630748902] plot(\x,{(--21.818129257845403-5.665443784839584*\x)/-1.3185195448197633});
\draw [color=cqcqcq,domain=-4.712094853111268:8.156577630748902] plot(\x,{(--19.69399678650052-3.17895765530317*\x)/-5.0481093114947955});
\draw [color=cqcqcq,domain=-4.712094853111268:8.156577630748902] plot(\x,{(--15.595940722712484-1.9909549277126777*\x)/-6.591835957108728});
\draw [color=cqcqcq,domain=-4.712094853111268:8.156577630748902] plot(\x,{(--22.83940915159945-3.8785018765083055*\x)/-2.7712711262542133});
\draw [color=cqcqcq,domain=-4.712094853111268:8.156577630748902] plot(\x,{(-11.74959087014205--5.412411851730033*\x)/-1.9520890608143557});
\draw [color=cqcqcq,domain=-4.712094853111268:8.156577630748902] plot(\x,{(--24.958110897186515-3.7593704758932*\x)/-0.9800846866173325});
\draw [color=cqcqcq,domain=-4.712094853111268:8.156577630748902] plot(\x,{(--23.08718429177135-5.681999939729836*\x)/-0.4831238765868111});
\draw [color=cqcqcq,domain=-4.712094853111268:8.156577630748902] plot(\x,{(-11.540747944423615-0.9564013074116611*\x)/3.5520747813729754});
\draw [color=cqcqcq,domain=-4.712094853111268:8.156577630748902] plot(\x,{(-18.647775718876495--0.7238299842444622*\x)/4.735701169805585});
\draw [color=cqcqcq,domain=-4.712094853111268:8.156577630748902] plot(\x,{(--22.19348133134985-3.622751634623669*\x)/-2.0682173108516295});
\draw [color=cqcqcq,domain=-4.712094853111268:8.156577630748902] plot(\x,{(-3.9017103708871623-3.895072183401471*\x)/8.760482358943197});
\draw [color=cqcqcq,domain=-4.712094853111268:8.156577630748902] plot(\x,{(-9.22781431480426-5.096629941971832*\x)/6.617436366576078});
\draw [color=cqcqcq,domain=-4.712094853111268:8.156577630748902] plot(\x,{(--21.23030597771169-6.124281858761327*\x)/4.7852490933197025});
\draw [color=cqcqcq,domain=-4.712094853111268:8.156577630748902] plot(\x,{(--13.983229107223542-6.126019849843834*\x)/7.169140006837228});
\begin{scriptsize}
\draw (-2.9152332656234927,0.8507903246401439) circle (1.0pt);
\draw (1.06,3.08) circle (1.0pt);
\draw (4.185326977539396,1.4361603223981485) circle (1.0pt);
\draw (5.770420411804214,-0.6230879827080192) circle (1.0pt);
\draw (0.15012831957960968,-3.289438309920027) circle (1.0pt);
\draw (-1.3238909135175247,3.0817379910825062) circle (1.0pt);
\draw (6.089783352700475,-2.1063069188269634) circle (1.0pt);
\draw (1.1095479235141181,-3.7681118430057894) circle (1.0pt);
\draw (2.5178973409027012,-2.817797812220428) circle (1.0pt);
\draw (0.4996314647178033,-4.800101200492752) circle (1.0pt);
\draw (5.105800704456993,-0.27600702705777147) circle (1.0pt);
\draw (2.8668074327196327,-4.229283462441436) circle (1.0pt);
\draw (4.440467812969361,-4.115966352036143) circle (1.0pt);
\draw (-0.821415545304514,-2.614042910420697) circle (1.0pt);
\draw (0.5899692474790508,-3.5297392412363005) circle (1.0pt);
\draw (5.8930163398675415,-0.9290053819537856) circle (1.0pt);
\draw (5.638078558973846,-0.35078158593313064) circle (1.0pt);
\draw (3.012089060814356,-2.332411851730033) circle (1.0pt);
\draw (7.069868039317807,1.6530635570662369) circle (1.0pt);
\draw (1.6315807733396361,-3.68832154038418) circle (1.0pt);
\draw (4.300360860231998,-1.0670951344754231) circle (1.0pt);
\draw (2.2234887175331006,-3.1069605129966686) circle (1.0pt);
\draw (3.396028831272468,-1.955313315381449) circle (1.0pt);
\draw (3.942605603245965,-1.4184757263309458) circle (1.0pt);
\draw (3.702203100952585,-4.245839617331688) circle (1.0pt);
\draw (5.845249093319703,-3.044281858761327) circle (1.0pt);
\end{scriptsize}
\end{tikzpicture}
\end{center}